\documentclass[a4paper, conference, 10pt]{IEEEtran}
\ifCLASSINFOpdf
\else
\fi
\usepackage[cmex10]{amsmath}
\usepackage{amssymb}
\usepackage{epsfig}
\usepackage{epstopdf,tikz,multirow}
\usepackage{verbatim}
\usepackage{pgfplots}
\usepackage{bigints}
\usetikzlibrary{patterns,decorations.markings}
\usepackage{amsthm,amssymb,graphicx,graphicx,multirow,amsmath,cite}
\usepackage{etoolbox}
\usepackage{caption}
\usepackage{subcaption,float,cite}

\newtheorem{lemma}{Lemma}{}
  \newtheorem{thm}{Theorem}

  \newtheorem{definition}[thm]{Definition}
  
\def\R{\mathbb{R}}
\def\P{\mathbb{P}} % for probabiltiy
\def\ie{{\em i.e.}}

\def\i{\mathbf{1}} % Indicator random variable
\def\N{\mathcal{N}}
\def\M{\mathcal{M}}
\def\S{\mathcal{S}}
\def\D{\mathcal{D}}

\usetikzlibrary{decorations.markings}

\begin{document}
%
% paper title
% can use linebreaks \\ within to get better formatting as desired
% Do not put math or special symbols in the title.
\title{Spatial CSMA: A Distributed Scheduling Algorithm for the SIR Model with Time-varying Channels}

\author{\IEEEauthorblockN{Peruru Subrahmanya Swamy, Radha Krishna Ganti, Krishna Jagannathan}
\IEEEauthorblockA{Department of Electrical Engineering,
IIT Madras,
Chennai, India 600036\\
\{p.swamy, rganti, krishnaj\}@ee.iitm.ac.in}
}

% author names and affiliations
% use a multiple column layout for up to three different
% affiliations

% conference papers do not typically use \thanks and this command
% is locked out in conference mode. If really needed, such as for
% the acknowledgment of grants, issue a \IEEEoverridecommandlockouts
% after \documentclass

% for over three affiliations, or if they all won't fit within the width
% of the page, use this alternative format:
% 
%\author{\IEEEauthorblockN{Michael Shell\IEEEauthorrefmark{1},
%Homer Simpson\IEEEauthorrefmark{2},
%James Kirk\IEEEauthorrefmark{3}, 
%Montgomery Scott\IEEEauthorrefmark{3} and
%Eldon Tyrell\IEEEauthorrefmark{4}}
%\IEEEauthorblockA{\IEEEauthorrefmark{1}School of Electrical and Computer Engineering\\
%Georgia Institute of Technology,
%Atlanta, Georgia 30332--0250\\ Email: see http://www.michaelshell.org/contact.html}
%\IEEEauthorblockA{\IEEEauthorrefmark{2}Twentieth Century Fox, Springfield, USA\\
%Email: homer@thesimpsons.com}
%\IEEEauthorblockA{\IEEEauthorrefmark{3}Starfleet Academy, San Francisco, California 96678-2391\\
%Telephone: (800) 555--1212, Fax: (888) 555--1212}
%\IEEEauthorblockA{\IEEEauthorrefmark{4}Tyrell Inc., 123 Replicant Street, Los Angeles, California 90210--4321}}

% use for special paper notices
%\IEEEspecialpapernotice{(Invited Paper)}

\IEEEoverridecommandlockouts
\IEEEpubid{\makebox[\columnwidth]{This work has been presented at NCC-2015, held at Mumbai, India.}
\hspace{\columnsep}\makebox[\columnwidth]{ }} 
% make the title area
\maketitle
\begin{abstract}
Recent work has shown that adaptive CSMA algorithms can achieve throughput optimality. However, these adaptive CSMA algorithms assume a rather simplistic model for the wireless medium. Specifically, the interference is typically modelled by a conflict graph, and the channels are assumed to be static. In this work, we propose a distributed and adaptive CSMA algorithm under a more realistic signal-to-interference ratio (SIR) based interference model, with time-varying channels. We prove that our algorithm is throughput optimal under this generalized model. Further, we augment our proposed algorithm by using a parallel update technique. Numerical results show that our algorithm outperforms the conflict graph based algorithms, in terms of supportable throughput and the rate of convergence to steady-state.
\end{abstract}

% As a general rule, do not put math, special symbols or citations
% in the abstract

% no keywords

% For peer review papers, you can put extra information on the cover
% page as needed:
% \ifCLASSOPTIONpeerreview
% \begin{center} \bfseries EDICS Category: 3-BBND \end{center}
% \fi
%
% For peerreview papers, this IEEEtran command inserts a page break and
% creates the second title. It will be ignored for other modes.

%\IEEEpeerreviewmaketitle

\section{Introduction}
A central problem in wireless networks is the design of efficient link scheduling algorithms in the presence of interference. In the design of scheduling algorithms, there are three key performance metrics of interest. The first among them is the achievable \emph{throughput region}. The throughput performance of a scheduling algorithm is characterized by the largest set of arrival rates under which the algorithm can stabilize the queues in the network. Secondly, the \emph{average delay} incurred by the packets in the queue should be small. The third metric of interest is the computational and communication \emph{complexity} involved in implementing the algorithm. Scheduling algorithms with low computational complexity and low communication overheads are preferable.
 
 \subsection{Related Work}
A large part of the existing literature on scheduling is based on the maximum weight scheduling algorithm \cite{tassiulas}, which is known to be throughput optimal under fairly general conditions. However, maximum weight scheduling generally requires solving an NP-hard problem during each scheduling instant, and is difficult to implement in practice. Further, it is not directly amenable to a distributed implementation. Several low complexity alternatives \cite{low_complex_alt1} have been proposed but they achieve only a fraction of the capacity region, and are hence not throughput optimal.
	
	On the other hand, there are simple random access techniques such as Aloha, CSMA (Carrier Sense Multiple Access) which can be implemented in a distributed manner. A distributed algorithm was developed in \cite{libin} to adaptively choose the CSMA parameters so as to achieve throughput optimality. Central to this algorithm is the so called \emph{Glauber dynamics}, which is a Monte Carlo Markov Chain sampling technique \cite{mixing_book}, \cite{bremaud}. Specifically, it is a Gibbs sampler \cite{bremaud} based algorithm.
	
	A main shortcoming of the existing papers on adaptive CSMA is that the results are derived based on rather simplistic models for the wireless channels and the interference. Typical modelling assumptions used include:
	\begin{itemize}
	\item[a.] \emph{Conflict graph interference model}: The interference is modelled by a conflict graph or protocol model \cite{libin}, where the transmissions from two links fail, if the links share an edge in the conflict graph. In reality however, the success or failure of a link depends on the aggregate interference from all the active links in the interference range. In other words, the complex nature of wireless interference is not adequately captured by a conflict graph. On the other hand, the SIR-based interference model can be used to overcome this limitation.
	\item[b.] \emph{Channel model}:  It is assumed that the wireless channel is either to be static (\ie,  not time-varying), or that the instantaneous CSI (Channel state information) at each time slot is available for scheduling collision free transmissions. However, wireless channels are seldom static due to fading, and the availability of CSI at each transmitter is not necessarily realistic in an adhoc setting.
	\end{itemize}	
	We present a brief summary of the assumptions made in the existing literature in the following table:
	\begin{center}
  \begin{tabular}{|c | c | c| c |c|}
  \hline
  Ref. & Interference & Channel  & CSI & Throughput\\  
     &model& model  &  & Optimality\\ 
    \hline
\cite{libin}, \cite{qcsma} & graph & static & - & \checkmark\\   
\cite{time_varying} & graph & varying & Inst. & \checkmark \\
\cite{greedy} & graph & varying & Inst. & \\
\cite{sinr_mimo} & SIR & varying & Inst.& \\
\cite{baccelli}& SIR & varying & Stat. & \\
This work& SIR & varying & Stat. & \checkmark\\
    \hline
    \end{tabular}
\end{center}

	\begin{itemize}
	\item \emph{Inst.} - Instantaneous channel gains are assumed to be known at each time slot.
	\item \emph{Stat.} - Channel statistics (such as average channel gains or  distribution) are assumed to be known.
	\end{itemize}
	
	A time-varying channel is considered between the transmitter of a link and its corresponding receiver in \cite{time_varying}. However, the channel gains between the interfering links are assumed to be static. In \cite{greedy}, time varying channels are considered among all the links, and the interference is modelled by a conflict graph. However, the algorithm \cite{greedy} can support only a fraction of the achievable rate region. A SIR model is considered in \cite{sinr_mimo} to a propose conservative algorithm that is suboptimal. An adaptive Aloha based algorithm is proposed in \cite{baccelli} under time-varying channels. However the algorithm can only maximize some utility functions and is not throughput optimal.
	\subsection{Our Contributions}
	 In this work, we consider a single-hop wireless network and propose a distributed scheduling algorithm.
	 \begin{itemize}
	 \item 	We consider time-varying channels among all the links in the network. Further, the interference is modelled using the SIR model which is more realistic.
	 \item  A key contribution of this paper is in the design of a Gibbs sampler \cite{bremaud} based throughput optimal scheduling algorithm (\emph{Algorithm 1}). In the algorithm we propose, each link only requires the average channel gains from its neighbouring links (defined later). In particular, instantaneous channel gains are not required, which makes our algorithm practical in a fast fading scenario, where the channel gains vary rapidly within a data slot.
     	  
	 \item We augment \emph{Algorithm 1}, which allows only single link updates, and propose \emph{Algorithm 2}, which performs parallel link updates and converges faster.
\end{itemize}	

The remainder of the paper is organised as follows. In Section \ref{network_model}, the network model is described. In Section \ref{rangeR}, the spatial CSMA algorithm is presented and its throughput optimality is proved. Numerical results are presented in Section \ref{simulations}, and we conclude in Section V.
\section{Network Model} \label{network_model}
We consider a single-hop ad-hoc wireless network, and model the links using  a  bipole model introduced in \cite{baccelli}. In a bipole model, each transmitter is associated with a receiver that is at a distance $R$ in some arbitrary direction.   A transmitter and its corresponding receiver is referred to as a link. We assume that there are $N$ links in the network. We use the set  $\mathcal{N}$ to denote all the links in the network.   We assume that the time is slotted. 

We assume that the link distance $R$ is much smaller than the distances of the transmitter and receiver to the other links.  With this assumption, we can think of links as points in the Euclidean space (The results in this paper are not limited by this assumption. The assumption is taken to keep the expressions simple).  Let $r_{ji}$ denote the distance between the links $i, j$. We consider a standard path-loss model $\|x\|^{-\alpha}, \alpha>2$.

{\noindent \em Channel model:} The small-scale fading (power) between any pair of nodes is modeled by a unit power  Rayleigh distribution and is assumed to be i.i.d across time and space.  The channel gain between the  transmitter of a link $i$ and the receiver in link $j$ is denoted by  $h_{ji}$. Since $h_{ij}$ is Rayleigh  distributed, $|h_{ij}|^2$ is exponentially distributed with unit mean. 

{\noindent \em Interference model:}
 A receiver successfully receives the packet of the corresponding transmitter if the received SIR is above a pre-determined threshold $T$. We consider interference limited networks, where the impact of thermal noise is negligible as compared to interference.  Suppose $\mathcal{M} \subset \mathcal{N}$, be the set of links that are transmitting in the current slot. The SIR of a link $i \in \mathcal{M}$ denoted by $\gamma_{i,\M}$ is given by, 
\begin{align*}
\gamma_{i,\mathcal{M}} =\frac{|h_{ii}|^2R^{-\alpha}}{I(\mathcal{M}\setminus\{i\})}.
\end{align*}
 Here $|h_{ii}|^2R^{-\alpha}$ is the received power at the receiver in link $i$  from its intended transmitter and  \[I(\mathcal{M}\setminus\{i\})=\sum\limits_{j \in \mathcal{M} \setminus \lbrace i \rbrace} |h_{ij}|^2r_{ij}^{-\alpha},\] is the interference power from other concurrent transmissions. 
 
{\noindent \em Queuing Dynamics:} 
 Each link has a separate arrival process and maintains its own buffer. $[a_i]_{i=1}^N$  denote the arrival rates of the links, $[q_i(t)]_{i=1}^N$ denote the queue lengths of the links in time slot $t$.   
 
 {\noindent \em Assumptions on channel state information:} We assume that each link knows the distances to its neighbours (defined later), the path loss exponent and the SIR threshold $T$. 
 
 We now compute the probability that a transmission is successful in the presence of interference. 
\subsection{Probability of successful link}
The probability of success for a link $i \in \mathcal{M}$ denoted by $\mu_i(\M)$ is 
\begin{align*}
\mu_i(\M)
&= \mathbb{P}\left(\gamma_{i,\mathcal{M}} \geq T\right), \\
&= \mathbb{P}\left( |h_{ii}|^2 \geq R^{\alpha}TI(\mathcal{M}\setminus\{i\}) \right), \\
&\stackrel{(a)}{=}\mathbb{E}_{\lbrace h_{ij} \rbrace} \exp\left(-TR^{\alpha}\sum\limits_{j \in \mathcal{M} \setminus \lbrace i \rbrace}|h_{ij}|^2r_{ij}^{-\alpha}\right), \\
&\stackrel{(b)}{=}\prod\limits_{j \in \mathcal{M} \setminus \lbrace i \rbrace} \mathbb{E}_{h_{ij}} \exp\left( -R^{\alpha}T |h_{ij}|^2r_{ij}^{-\alpha} \right), \\
&\stackrel{(c)}{=}\prod\limits_{j \in \mathcal{M} \setminus \lbrace i \rbrace} \frac{1}{1+ {\left( \frac{R}{r_{ij}} \right)}^{\alpha}T},  \hspace{10mm} \forall i \in \mathcal{M},
\end{align*}
where $(a)$ and $(c)$ follow from the exponential distribution of $|h_{ii}|^2$, $|h_{ij}|^2$ and $(b)$ follows from the independence of the fading variables. 
Let  $f(r_{ij}):=\frac{1}{1+ {\left( \frac{R}{r_{ij}} \right)}^{\alpha}T}.$ Then the probability of success can be written as,
\begin{align}
\mu_i(\M)&= \prod\limits_{j \in \mathcal{M} \setminus \lbrace i \rbrace} f(r_{ij}), \quad \forall i \in \mathcal{M}. \label{productform}
\end{align}
For convenience,  the probability of success is set to zero for the links that are not in the currently active set $\mathcal{M}$, \ie, \;
$\mu_i(\M)=0, \; \; \forall i \notin \mathcal{M}$. 
Note that \eqref{productform} is calculated, assuming all the active links in the network can contribute to the interference of a receiver.  However, from the studies on statistical distribution of co-channel interference, the aggregate interference from the links beyond a certain distance can be safely neglected \cite{radius_approx1 , radius_approx2}.  The radius beyond which the interference can be neglected is referred to as  close-in radius, and is denoted by $R_I$.  Hence for a link $i$, the interference from the active links outside a ball of radius $R_I$ around $i$ can be neglected. 
Let $\N_i$ denote the set of links that are potential interferers of link $i$, \ie, the set of links within the ball of radius $R_I$ around link $i$. The links in $\N_i$ are referred to as the \emph{neighbours} of link $i$. Thus, from \eqref{productform} the probability of success is,
\begin{align}
\mu_i(\M)&= \prod\limits_{j \in \M_i   } f(r_{ij}), \quad \forall i \in \mathcal{M}, \label{productform1}
\end{align}
where  $\M_i := \N_i \cap \M$ is set of active links that are within the close-in radius of link $i$. Also note that, if none of the potential interferers of a link $i$ are active, then it succeeds with probability one, \ie, 
\begin{align*}
\mu_i(\M)&=1, \quad  \text{if\; } \M \cap \N_i = \emptyset.
\end{align*}

From \eqref{productform1}, we can observe that the probability of success of a link depends only on the distances from its \emph{active} neighbours $\M_i$.  This  allows for  computation of $\mu_i(\M)$ by  a simple neighbour discovery algorithm \cite{neigh_discovery}. 

We now characterize the capacity region in terms of the link success probabilities $\mu_i(\M)$.
\subsection{Capacity Region}
Every subset of the links $\mathcal{M} \subset \mathcal{N}$,  is associated with a $N$-dimensional vector $\mu(\M)= [\mu_i(\M)]_{i \in \N}$ whose $i$-th element correspond to the probability of success of the link $i$ (when $\mathcal{M}$ is the set of links that are transmitting). $\mu(\M)$ can also be interpreted as the long-term rates that can be supported when the subset $\M$ is active. We refer to these vectors as rate vectors.

The capacity region of the network is the set of all the arrival rate vectors for which there exists a scheduling algorithm that can stabilize the queues. It is known that the capacity region is given by
\begin{align*}
\Lambda&=\lbrace {a \in \mathbb{R}^N_+} \; | \;  \; \exists \epsilon>0,  \; a(1+\epsilon) \in \mathcal{C}o\left(\mu\right) \rbrace, 
\end{align*}
where, $\mathcal{C}o(\mu)$ is the convex hull of $\lbrace \mu(\mathcal{M}) \rbrace_{\mathcal{M} \subset \mathcal{N}}$.

An arrival rate vector $y\in \R^n$ is said to be  feasible if $y  \in \Lambda$. A scheduling algorithm is said to be $\emph{throughput optimal}$, if the algorithm can stabilize the network for any feasible arrival rate. A maximum weight scheduling algorithm is known to be throughput optimal. In each time slot, the algorithm picks the schedule, $\M(t)=\underset{\M \subset \N}{\arg\max}  \sum\limits_{j \in \M} \mu_j\left(\M\right)  q_j(t)$.\\
Some of the notations used so far, are summarized below.
  \begin{center}
  \begin{tabular}{{|r l|}}
    \hline
$\N-$ & \emph{Set of all the links in the network}\\
$\M(t)-$& \emph{Set of links that are active in slot $t$}\\
$\N_i-$& \emph{Set of potential interferers of link $i$.} \\
$\M_i(t)-$& \emph{Set of  active interferers of link $i$ in  slot $t$.}\\
$\mu_i(\M)-$ &\emph{Rate of link $i$ when the set of active} \\
\; & \emph{  links is $\M$.}\\
    \hline
    \end{tabular}
\end{center}

\section{Spatial CSMA} \label{rangeR}
In this Section, our distributed algorithm, \emph{Spatial CSMA} is presented and its throughput optimality is proved. The key idea is to sample subsets (of links) so that sampled subsets provide a good approximation to the Maximum weight algorithm \cite{qcsma,libin}. 
Let $g(x)$ be a real valued function of queue length.  The details of the function $g(x)$ are discussed subsequently. 
\noindent\rule[0.5ex]{\linewidth}{0.5pt}
\emph{Algorithm1}: \textbf{Spatial CSMA} \\
\noindent\rule[0.5ex]{\linewidth}{0.5pt}
\textbf{Intialization:} Each link  $i \in \mathcal{N}$ pre-computes $f_{ij}:=f(r_{ij})$ for all its neighbours $j \in \N_i$ .\\
\textbf{Control slot:}
\begin{itemize}
\item \emph{Decision schedule-} 
A link $i \in \N$, is picked uniformly at random.
\item \emph{Neighbour discovery-} Each link $j \in \{i\} \cup \N_i$  executes a neighbour discovery \cite{neigh_discovery} algorithm to compute  the set of its active interferes in the previous slot, \ie, $\M_j(t-1)$.
\item \emph{Inactive weights-} Each link $j \in \N_i$ computes $\mu_j(\M(t-1) \setminus \lbrace i \rbrace) $ from \eqref{productform1} and subsequently computes the inactive weight 
\begin{align}
w_j^0 :=g\left( q_j(t) \right) \mu_j(\M(t-1) \setminus \lbrace i \rbrace).  \label{weight}
\end{align}

%\item \emph{Message passing-} Each link  $j \in \N_i$, passes its inactive weight $w_j^0$ computed in \eqref{weight} to link $i$.
\item \emph{Active weights-} Link $i$ obtains the inactive weights from its neighbours and computes the active weights as defined below.
\begin{align*}
w_j^1 &:= w_j^0 f_{ij}, \quad  \forall j \in \N_i ,\\
w_i^1 &:=g\left( q_i(t) \right) \mu_i(\M(t-1) \cup \lbrace i \rbrace ).
\end{align*} 
\item \emph{Update Probability-} Link $i$ computes its update probability $p(t)$ as,
\begin{align}
 p(t)= \frac{\exp\left({w_i^1}\right)}{\exp\left(\sum\limits_{j \in \M_i(t-1)} \left(w_j^0 - w_j^1\right)\right)+ \exp\left({w_i^1}\right)}. \label{update_prob_eq}
\end{align}
 Link $i$ chooses to transmit with probability $p(t)$  and chooses not to transmit with probability $1-p(t)$, \ie,
 \[
\M(t)= \left\{\begin{array}{ll}
 \M(t-1) \cup \lbrace i \rbrace & \text{w.p. \; \; }  p(t), \\
 \M(t-1) \setminus \{i \} & \text{w.p.\; \;}  1-p(t).
 \end{array}
 \right. \]
\end{itemize}
\textbf{Data slot:} In the data slot, all the links $j \in \M(t)$ will transmit.
\noindent\rule[0.5ex]{\linewidth}{0.5pt}

In \emph{Algorithm1}, each time slot is divided into a control slot and a data slot. In the control slot, a link $i$ is chosen at random (the implementation of this step is discussed later), and only  this link is allowed to change its status (on/off) in this time slot. All other links will retain their status of the previous time slot. Link $i$ and its neighbours execute a neighbour discovery algorithm to identify all their active neighbours.   For example,  the  compressed neighbour discovery scheme \cite{neigh_discovery}  is a fast and efficient neighbour discovery algorithm which jointly detects all the active neighbours by allowing them to simultaneously report their identity. 

All the neighbours of the link $i$, use their neighbourhood information $\M_j(t-1)$ computed in the previous step to calculate their rate vectors. Note that, all the links in $\M_j(t-1)$ retain their status except for the possible change of the status for link $i$. Hence, to account for this possible change of the status of link $i$,  we define two sets of weights namely \emph{inactive weights} and \emph{active weights}. The contribution of interference from link $i$ is excluded for computing the inactive weights but included for computing the \emph{active weights}. Link $i$ uses these weights to compute its update probability $p(t)$, and updates its status accordingly. In the data slot, all the active links transmit.

\subsection{Throughput Optimality}

\begin{lemma} \label{lemma_dist}
If the queue lengths are fixed at $q= [q_i]_{i =1}^{N}$, then \emph{Algorithm1} corresponds to a Glauber dynamics Markov chain on the subsets  $\M \subset \N$ with a stationary distribution given by,
\begin{align}
\Pi(\M)&= \frac{1}{Z} \exp\left(\sum\limits_{j \in \M} \mu_j\left(\M\right)  g(q_j)\right), \quad  \forall \M \subset \N ,\label{stat_dist}
\end{align}
where $Z$ is the  normalizing constant.
\end{lemma}
\begin{proof}
Proof can be found in Appendix-\ref{proof_dist}
\end{proof}

If the queue lengths were indeed fixed (say at $q$) as required in Lemma \ref{lemma_dist},  \emph{Algorithm1}  provides a good approximation for  the  maximum weight scheduler \cite{qcsma}. This can be observed from \eqref{stat_dist},  as the stationary distribution $\Pi$ on the set of subsets,  places the largest mass on the set $\mathcal{M}$ that maximizes  $\sum\limits_{j \in \M} \mu_j\left(\M\right)  g(q_j)$,  which is precisely the max-weight scheduler except for $q_j$ being replaced by $g(q_j)$. However, this replacement can be justified if an appropriate function $g(x)$ is chosen \cite{stable_policies}.

  Lemma \ref{lemma_dist} assumes that the queue lengths are fixed. However, the queue lengths are time-varying. Moreover, the time required for the Glauber dynamics to reach steady-state can be very long in general to assume that the queue lengths do not change. However, if appropriate slowly varying functions like $\log(0.1x),  \log \log (x+e)$ are used as $g(x)$, it can be shown \cite{qcsma},\cite{dshah} that \emph{Algorithm1} does approximate the maximum weight scheduler in each time slot with a high probability and is hence throughput optimal.
\begin{lemma}
If $g(x) = \log(0.1x)$ or $\log\log(x+e)$, the proposed  spatial CSMA algorithm is throughput optimal. 
\end{lemma}
\begin{proof}
Follows from our Lemma \ref{lemma_dist}, and Theorem 1, Proposition 2 in \cite{qcsma}.
\end{proof}

\emph{Remarks:} While the techniques used are standard, the key contribution of this paper is the application of these techniques to design a throughput optimal scheduling algorithm for the SIR model with time-varying channels. 

Although \emph{Algorithm1} is proposed for a Rayleigh fading model, it can be easily extended to other fading models without any additional effort.

In \cite{time_varying}, CSMA  algorithm is considered on a conflict graph with time-varying link capacities. The authors of \cite{time_varying} show that the back-off parameter should have a exponential form of the channel gain. They obtain this by solving a maximum entropy problem. However, as we see from  Lemma \ref{lemma_dist}, the  exponential form follows naturally from the max-weight formulation. 

\begin{comment}
\item Multiple update constraints are prescribed and a decision schedule algorithm is designed. This multiple update rule has better perforamance than single.. include simulations comparing them
\end{comment}

\begin{comment}
\begin{lemma} \label{lemma_optimal}
Spatial CSMA Alogrithm statibilizes the network for any arrival rate vector $a \in \Lambda$.
\end{lemma}
\begin{proof}
Proof of this lemma directly follows from the above lemma and theorem $1$ in \cite{qcsma}. we omit the proof for brevity.
\end{proof}
\end{comment}
\begin{figure*}
\centering
\begin{subfigure}{.48\textwidth}
  \centering
  \begin{tikzpicture}[scale=0.8]
\draw[help lines] (0,0) grid (6,6);
\begin{axis}[%
width=6cm,
height=6cm,
scale only axis,
xmin=0,
xmax=0.35,
xlabel={Arrival rate},
ymin=0,
ymax=320,
ylabel={Average queue length},
legend style={at={(0.05,0.7)},anchor=south west,draw=black,fill=white,legend cell align=left}
]

%\draw[step=1cm, black,thick] (-2,-2) grid (6,6);
%\draw[help lines] (0,0) grid (2,3);

\addplot [
color=black,
solid,
mark=star,
mark options={solid}
]
table[row sep=crcr]{
.05 1.57\\ 
.1 2.9651\\
.15  5.09\\ 
.2 9.28\\ 
.21 10.27\\ 
.22 10.95\\ 
.23 12.43\\
.24  13.79\\ 
.25 17.15\\ 
.26 19.07\\ 
.27 25.3\\
.28 25.5\\ 
.29 35.42\\ 
.3 52.9\\ 
.31 79.1\\ 
.32 103.3\\ 
.33 235.2\\    
};
\addlegendentry{ SIR model};

\addplot [
color=red,
solid,
mark=*,
mark options={solid}
]
table[row sep=crcr]{
.05 4.35\\ 
.1 9.26\\ 
.15  20.85\\ 
.2 76.39\\
.21  170.05\\
.22 241.6\\ 
.23 313.7\\
};
\addlegendentry{Conflict graph};

\end{axis}
\end{tikzpicture}%
\caption{ Throughput Comparison}
\label{fig:arrival_rate}
\end{subfigure}%
\begin{subfigure}{.48\textwidth}
\includegraphics[scale=0.4]{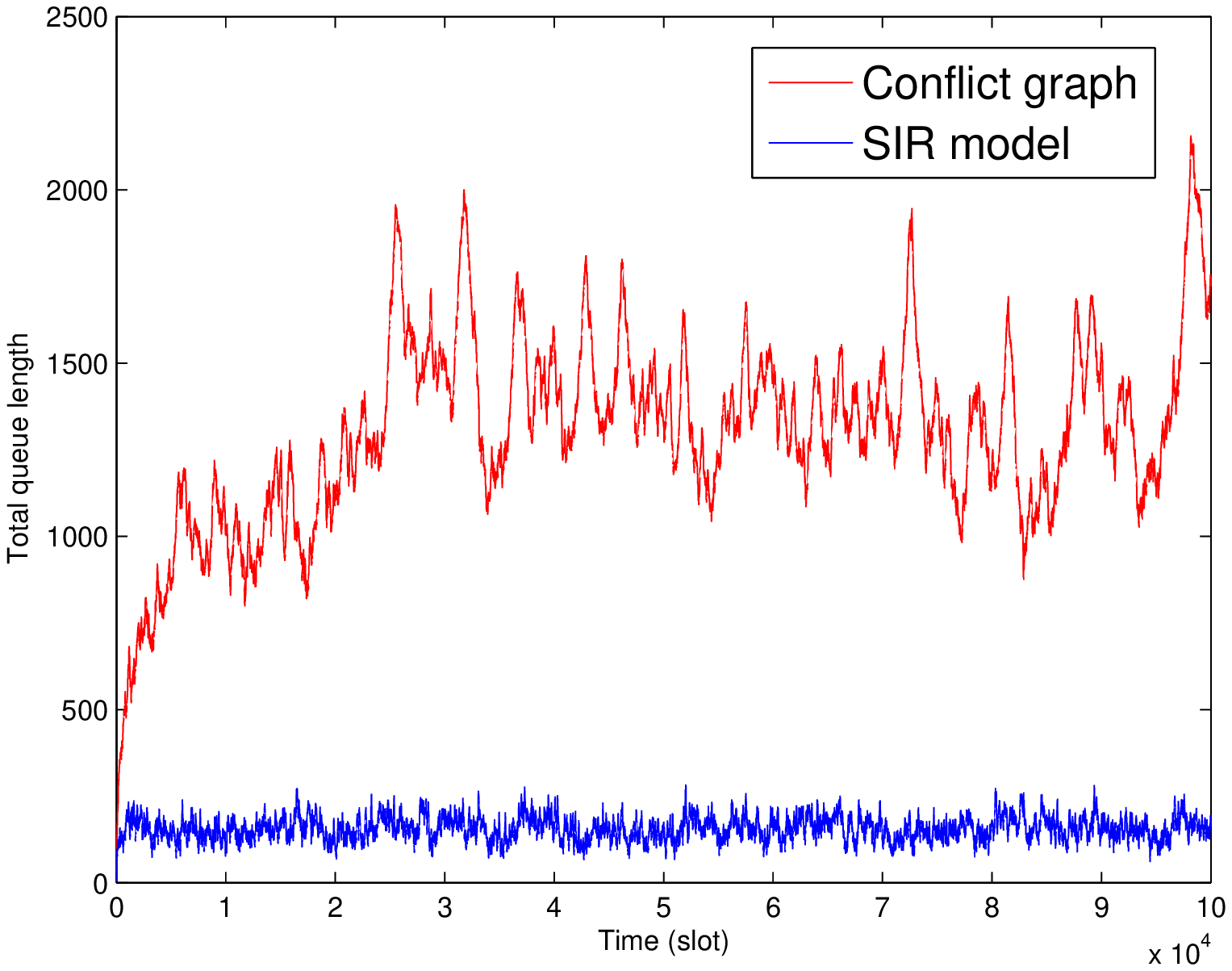} 
\caption{Convergence rate comparison}
\label{q_comparison}
\end{subfigure}
\caption{Numerical Results}
\label{fig:test}
\end{figure*}
\subsection{Parallel Updates}
In the control slot of \emph{Algorithm1}, it is assumed that a link (decision schedule) can be randomly picked at each time slot to update its status. However, the algorithm does not explicitly describe how to implement that step. Moreover, only one link is allowed to update its status in a given time slot. In \emph{Algorithm2}, we relax the limitation of single-update, and provide a distributed algorithm to pick a decision schedule. The limitation of single-update can be relaxed by considering a block (parallel update) Gibbs sampler based algorithm, which allows for parallel updates and also converges faster. However, to ensure a distributed implementation of the (parallel update) Gibbs sampler, the set of links that can do parallel updates has to satisfy the following constraint. (See Lemma \ref{lemma_multiple} for a formal proof.) \\
 \emph{If link $i$ updates its status in a given slot, all the links whose current status information is being used in the computation of the update probability $p(t)$ of link $i$, cannot update in the same slot}. \\
 A set of links, which satisfy the above constraint is referred to as a decision schedule. The formal definition is as follows:

\begin{definition} \textbf{Decision Schedule}\\
 A set of links $\D \subset \N$,  is said to be a decision schedule if,
\begin{align}
\N_i \cup \left(\bigcup\limits_{j \in \N_i } \N_j \setminus \lbrace i \rbrace\right) &\subset \N \setminus \D,  & \forall i \in \D. \label{def_dec}
\end{align}
\end{definition}
The intuition for this definition is as follows. From \eqref{update_prob_eq}, one can observe that the update probability $p(t)$ of a link $i$, depends only on the weights of active links in $\N_i$. Further, the weight of each link $j \in \N_i$, depends on the the status of its neighbours $\N_j$. (See the computation of incative weights in \emph{Algorithm1}.) Hence, the required constraint translates to \eqref{def_dec}.

\emph{Remark:} This constraint is only on the set of links that can do parallel updates in a given slot. However, there are no hard constraints on the set of links than can transmit in a given slot. This is a key difference of this model compared to conflict graph based model.\\
Generating a decision schedule $\D$ can be done in two steps.\\
\emph{Step 1:} Generate a subset of links  $\S$, such that no two links  in $\S$ are within the close-in radius of each other. \\
\emph{Step 2:} Initialize $\D$ to $\S$, and update $\D$ by removing some links from $\D$ as follows. Each link $k \notin \S$ checks if any of its neighbours are in $\S$. If more than one of its neighbours are present in $\S$,  then the neighbours are removed from $\D$.

The first step ensures $\N_i \subset \N \setminus \D$, while the second step ensures $\left(\bigcup\limits_{j \in \N_i } \N_j \setminus \lbrace i \rbrace\right) \subset \N \setminus \D$ so that \eqref{def_dec} is satisfied.

In \cite{qcsma}, a distributed algorithm is suggested for generating the subset $\S$ (\emph{step 1}). We extend that algorithm to generate the subset $\D$. For the sake of completeness, we present the \emph{step 1} from \cite{qcsma}. In \emph{Algorithm2}, the control slot is divided into $(W+2)$ control mini-slots for some $W \geq 2$. (This bound on $W$ is to ensure that each link has a non zero probability of being selected in the decision schedule.) In each time slot, all the links in the network will execute \emph{Algorithm2} independently. 

\noindent\rule[0.5ex]{\linewidth}{0.5pt}
\emph{Algorithm2}: \textbf{Decision Schedule Algorithm (at link $i$)} \\
\noindent\rule[0.5ex]{\linewidth}{0.5pt}
\emph{Step1:}\textbf{ Generating $\S$}
\begin{itemize}
\item Link $i$ selects a random (integer) backoff time $T_i$ uniformly in $\left[0, W-1\right]$ and waits for $T_i$ control mini-slots.
\item If link $i$ hears an INTENT message from a link in $\N_i$ before the $\left(T_i + 1 \right)$th control mini-slot, $i$ will not be included in $\S$ and will not transmit an INTENT message.
\item If link $i$ does not hear an INTENT message from any link in $\N_i$, before the $\left(T_i + 1 \right)$th control mini-slot, it will send (broadcast) an INTENT message to all links in $\N_i$ at the beginning of the $\left(T_i + 1 \right)$th control mini-slot.
\begin{itemize}
\item If there is a collision (\ie, if there is another link in $\N_i$ transmitting an INTENT message in the same mini-slot), link $i$ will not be included in $\S$.
\item If there is no collision, link $i$ will be included in $\S$.
\end{itemize}
\end{itemize}
\emph{Step2: \textbf{Generating $\D$ from $\S$}}
\begin{itemize}
\item If link $i \in \S$, it sends (broadcasts) an INTENT message to all links in $\N_i$ at the beginning of $(W+1)$th control mini-slot.
\item If link $i \notin \S$, it senses the channel for a possible collision (\ie, if more than one of its neighbours are in $\S$, then all of them send INTENT messages which result in a collision) in $(W+1)$th mini-slot. 
\begin{itemize}
\item If there is a collision, link $i$ will broadcast a DETECT message to all its neighbours in $(W+2)$th control mini-slot.
\end{itemize}
\item If link $i \in \S$ and it doesn't hear a DETECT message in $(W+2)$th control mini-slot, it will be included in $\D$.
\end{itemize}
\noindent\rule[0.5ex]{\linewidth}{0.5pt}
\begin{lemma} \label{lemma_multiple}
If all the links in a decision schedule $\D$ selected from \emph{Algorithm2}, simultaneously update their schedules using  \emph{Algorithm1}, then the stationary distribution of the resulting (parallel update) Glauber dynamics is given by \eqref{stat_dist}.
\end{lemma}
\begin{proof}
Proof can be found in Appendix \ref{proof_multiple}.
\end{proof}

\section{Numerical Results} \label{simulations}
In this Section, we evaluate the performance of the spatial CSMA algorithm. The results are compared to conflict graph based CSMA. This comparison requires the generation of an equivalent conflict graph for a given set of locations (of the links) as described below. 

\emph{Construction of conflict graph:} In a conflict graph based interference model, each link in the network is represented by a vertex in a graph. Two vertices are connected by an edge, if their concurrent transmissions can possibly end up in a collision. Concurrent transmissions from any two links that are with in the close-in radius of each other can be unsuccessful (depending on the channel conditions). Hence, a pair of vertices are connected by an edge if they are in the close-in distance of each other.
\subsection{Simulation settings}
We consider a two dimensional square plane with side length $13$.  A homogeneous Poisson point process of density $0.1$ is generated. The generated points correspond to the locations of the transmitters. Each transmitter has its receiver at a distance of $0.25$ in a random direction. The path loss exponent $\alpha$ is set to $2.5$, the close-in radius $R_I$ is set to $4$, and the threshold SIR is set to $17$ dB. The function $g(x)$ (used in \emph{Algorithm1}) is set to $\log(0.1x)$. In each time slot, the channel gains corresponding to unit power Rayleigh distribution are generated. 
\subsection{Throughput performance }
In Figure \ref{fig:arrival_rate}, we illustrate the throughput performance of \emph{Algorithm1}, by plotting the average queue lengths to see for which arrival rates the system is stable. If the algorithm cannot stabilize the network for a given arrival rate, the queue length blows up.  We consider homogeneous arrival rates for all the links. 
It can be observed that the SIR based algorithm supports a larger set of arrival rates compared to the graph based algorithm. This is because, in a conflict graph model, concurrent transmissions from two neighbouring links are strictly prohibited irrespective of the exact distance between them. However, in SIR model, the links make a better choice by considering the exact distances from its neighbouring links (thereby taking into account the severity of interference) while computing the update probabilities.

\subsection{Convergence rate}
 In Figure \ref{q_comparison}, we compare the convergence rate of the spatial CSMA with the graph model by plotting the total queue evolution as a function of time. We consider a homogeneous arrival rate of $0.2$ which is in the stable region of both these models. It can be observed that the queue reaches steady state much faster in the SIR model as compared to the conflict graph model.

\section{Concluding Remarks}
In this paper, we considered the SIR model with time-varying channels, and proposed a distributed CSMA algorithm. We further proved that the proposed algorithm is throughput optimal. We also proposed a parallel update algorithm with a better convergence rate. Using simulations, we observed that the SIR model supports a larger set of arrival rates, and converges much faster than the conflict graph based model.
\begin{comment}
\begin{itemize}

\item The access probabilities used in the algorithm doesn't have an impact on the stationary distribution of the Markov chain. But it may have any effect on the convergence rate. so how do we choose those access probabilities in a better way?

\item Should I mention, that if we take $\mu_j(1,0,0,0...) =1$ and $ \mu_(1,any othercomb)=-\infty$ and $\mu_(0,any other combination)=0$ gives conflict graph model.

\end{itemize}
\end{comment}

\section{Appendix}
\subsection{Proof of Lemma - \ref{lemma_dist}} \label{proof_dist}
The Glauber dynamics (section 3.3.2 of \cite{mixing_book}) corresponding to the distribution $\Pi$, is a reversible Markov chain with state space $\lbrace \M \;  | \; \M \subset \N \rbrace$, and stationary distribution $\Pi$. The transition probabilities of that Markov chain are described here. From a given state $\M(t-1)$, the chain moves to a new state as follows. A link $i$ is chosen uniformly at random from $\N$ and a new state is chosen according to the measure $\Pi$ conditioned on the set of states in which status (on/off state) of all the links other than link $i$ remain the same as in $\M(t-1).$ In other words, the chain can only move to $\M(t-1) \cup \lbrace i \rbrace$ or $\M(t-1) \setminus \lbrace i \rbrace$ according to the conditional distribution
\begin{align}
\Pi(\M(t) \; | \; \M(t) \in \lbrace \M(t-1) \cup \lbrace i \rbrace , \M(t-1) \setminus \lbrace i \rbrace \rbrace). \label{update_eq}
\end{align}
From \eqref{stat_dist}, it can be easily verified that the update probability $p(t)$ in \emph{Algorithm1} correponds to the conditional distribution in (\ref{update_eq}). Hence, \emph{Algorithm1} corresponds to a Glauber dynamics Markov chain with stationary distribution given by \eqref{stat_dist}.

\subsection{Proof of Lemma - \ref{lemma_multiple}} \label{proof_multiple}
Let us represent a state $\M$ of the Glauber dynamics Markov chain with a $N-$ dimensional binary vector $\sigma$ whose elements are given by $\sigma_i=\i\lbrace i \in \M \rbrace, \; \; {i \in \N }.$ Thus, the equivalent state space of the Markov chain is $\lbrace \sigma\;|\;\sigma \in \lbrace 0,1\rbrace^N \rbrace$. $\sigma(t)$ can be shown as a Random field (chapter 7 in \cite{bremaud}) on $\N$ with $\lbrace\sigma_i(t)\rbrace_{i \in \N}$ being the underlying random variables.  We construct an undirected graph $G(V,E)$ with $V=\N$ and edges given by the following definition of the neighbourhood on $G$.
\begin{align*}
N_G(i):= \left(\bigcup\limits_{k \in \N_i} \N_k \setminus \lbrace i \rbrace\right) \cup \N_i,
\end{align*}
where, $N_G(i)$ is the set of neighbouring vertices of vertex $i$ in the graph $G$. Let us denote, $\sigma(\S)=[\sigma_i]_{i \in \S}$ for $\S \subset \N$. Observe that the conditional distribution, $\sigma_i(t)$ given $\sigma(\N \setminus \lbrace i \rbrace)(t)$ corresponds to the distribution in \eqref{update_eq} (which corresponds to the update probability $p(t)$ as discussed in the proof of lemma \ref{lemma_dist}). Hence, 
\begin{align}
\P(\sigma_i(t)=1 \; | \; \sigma(\N \setminus \lbrace i \rbrace)(t))&= p(t). \label{lhs}
\end{align}
 From \eqref{update_prob_eq}, one can observe that the update probability $p(t)$, depends only on the weights of active links in $\N_i$. Further, the weight of each link $k \in \N_i$ given by \eqref{weight}, depends on the the status of its neighbours $\N_k$.  In other words, the update probability $p(t)$ of link $i$, depends only on the status of the links in $N_G(i)$. Using this observation in \eqref{lhs} gives,
\begin{align}
\P\left(\sigma_i(t) \; | \; \sigma(\N \setminus \lbrace i \rbrace)(t)\right)&= \P\left(\sigma_i(t) \; | \; \sigma(N_G(i))(t)\right). \label{Markov_prop}
\end{align}
In other words, the random variable $\sigma_i(t)$ is independent of all other random variables given the random variables corresponding to its neighbours in $G$, \ie, $\sigma(t)$ is a Markov random field \cite{bremaud} with respect to $G$.

In a parallel update Glauber dynamics corresponding to $\Pi$, if $\D$ is the set of links that are selected for update, then the update rule  is specified by the following joint law \cite{bremaud}, 
\begin{align}
\P(\sigma(\D)(t) \; | \; \sigma(\N \setminus \D)(t))&=\Pi(\sigma(t) \; | \; \sigma(\N \setminus \D)(t)). \label{block_update}
\end{align}
From \eqref{def_dec},  $i \in \D \Rightarrow N_G(i) \subset \N \setminus \D$. Using this property along with \eqref{Markov_prop}, we can write \eqref{block_update} as a product of single-site update probabilities as follows.
\begin{align*}
\P(\sigma(\D)(t) \; | \; \sigma(\N \setminus \D)(t))&=\prod\limits_{i \in \D(t)} \P\left(\sigma_i(t) \; | \; \sigma(N_G(i))(t)\right).
\end{align*}
In other words, we can use the same update rule as in \emph{Algorithm1} for all the links in $\D(t)$ simultaneously and still converge to the required stationary distribution given by \eqref{stat_dist}.

\begin{comment}

\end{comment}

\bibliographystyle{IEEEtran}
\bibliography{myreferences_ncc2015}
% that's all folks
\end{document}